\documentclass[superscriptaddress,showpacs,nofootinbib,tightenlines,twocolumn]{revtex4-1}
\usepackage{amsmath,verbatim,latexsym,amssymb,graphicx,indentfirst,mathrsfs,mathtools,amsthm,bbm,bm,cancel,hyperref}
\usepackage[font=small,labelfont=bf,format=plain,justification=raggedright,singlelinecheck=false]{caption}

\newtheoremstyle{indented}{5pt}{3pt}{\addtolength{\leftskip}{3.5em}}{}{\bfseries}{.}{.5em}{}
\newtheorem{theorem}{Theorem}

\newtheorem{definition}{Definition}

\makeatletter
\newcommand\footnoteref[1]{\protected@xdef\@thefnmark{\ref{#1}}\@footnotemark}
\makeatother

\usepackage{soul}

\usepackage{color}
\usepackage[dvipsnames]{xcolor}
\usepackage[normalem]{ulem}

\def\kB{k_{\rm B}}

\def\est{{\rm est}}

\newcommand{\fwd}{\mathrm{fwd}}
\newcommand{\rev}{\mathrm{rev}}

\newcommand{\caphead}[1]{{\bf #1}}

\begin{document}

\title{Number of trials required to estimate a free-energy difference, \protect\\  using fluctuation
relations}

 \author{Nicole~Yunger~Halpern\footnote{E-mail: nicoleyh@caltech.edu}}
 \affiliation{Institute for Quantum Information and Matter, Caltech, Pasadena, CA 91125, USA}
 \author{Christopher~Jarzynski}
 \affiliation{Department of Chemistry and Biochemistry, University of Maryland, College Park, MD 20742, USA}
 \affiliation{Institute for Physical Science and Technology, University of Maryland, College Park, MD 20742, USA}
\date{\today}


%
%
%
%
 \begin{abstract}
The difference $\Delta F$ between free energies
has applications in biology, chemistry, and pharmacology.
The value of $\Delta F$ can be estimated 
from experiments or simulations,
via fluctuation theorems developed in statistical mechanics.
Calculating the error in a $(\Delta F)$-estimate is difficult.
Worse, atypical trials dominate estimates.
How many trials one should perform
was estimated roughly in [Jarzynski, \emph{Phys. Rev. E} \textbf{73}, 046105 (2006)].
 We enhance the approximation with information-theoretic strategies:
 We quantify ``dominance'' with a tolerance parameter 
 chosen by the experimenter or simulator.
 We bound the number of trials one should expect to perform,
 using the order-$\infty$ R\'enyi entropy.
 The bound can be estimated if one implements the ``good practice'' 
 of bidirectionality, known to improve estimates of $\Delta F$.
 Estimating $\Delta F$ from this number of trials
 leads to an error that we bound approximately.
Numerical experiments on a weakly interacting dilute classical gas
 support our analytical calculations.

 \end{abstract}
 \maketitle

%
%
%
%
The numerical estimation of free-energy differences is an active area of research, 
having applications to chemistry, microbiology, pharmacology, and other fields.
Fluctuation relations can be used to estimate equilibrium free-energy differences $\Delta F$ from nonequilibrium experimental and simulation data. 
One repeatedly measures the amount $W$ of work extracted from, or performed on, a system during an experiment or simulation.
Fluctuation relations express the value of $\Delta F$ in terms of averages over infinitely many trials. 
Finitely many trials are performed in practice, 
introducing errors into estimates of $\Delta F$.
Efforts to quantify these errors, and to promote ``good practices'' in estimating $\Delta F$, 
have been initiated (e.g.,~\cite{GoodPractices,RohwerAT14,LuKofke_JChemPhys_01_I,LuKofke_JChemPhys_01_II,GoreRB03,WuKofke_JChemPhys_04,WuKofke_PRE_04,WuKofke_JChemPhys_2005_I,WuKofke_JChemPhys_2005_II,Kofke06,HahnT09,KimTalkner12}).

How many trials should one perform to estimate $\Delta F$ reliably? 
The work $W$ extracted from a system 
is a random variable that assumes different values in different trials.
Typical trials involve $W$-values that contribute little to the averages being estimated.
\emph{Dominant} $W$-values, which largely determine the averages, 
characterize few trials~\cite{RareEvents}.
Until observing a dominant $W$-value, one cannot estimate $\Delta F$ with reasonable accuracy.
The probability that some trial will involve a dominant $W$-value 
determines the number $N$ of trials one should expect to perform.

A rough estimate of $N$ was provided in~\cite{RareEvents}.
In this paper, we enhance the estimate's precision.
First, we introduce fluctuation relations and \emph{one-shot information theory},
a mathematical toolkit for quantifying efficiencies at small scales.
Next, we quantify dominance in terms of a tolerance parameter $w^\delta$.
We bound the number $N_\delta$ of trials expected to be required 
to observe a dominant work value.
This bound depends on the thermal order-$\infty$ R\'enyi entropy $H^\beta_\infty$, 
a quantity inspired by one-shot information theory~\cite{YungerHalpernGDV15}.
The bound can be estimated during an implementation of 
the ``bidirectionality good practice'' recommended in~\cite{GoodPractices}.
Finally, we approximately bound the error 
in a $(\Delta F)$-estimate inferred from $N_\delta$ trials.
A weakly interacting dilute classical gas~\cite{CrooksJarz} illustrates our analytical results.

%
%
%
%
\textbf{Technical introduction---}Let us introduce nonequilibrium fluctuation relations and 
the thermal order-$\infty$ R\'enyi entropy $H^\beta_\infty$.

\emph{Nonequilibrium fluctuation relations---}Nonequilibrium fluctuation relations govern statistical mechanical systems arbitrarily far from equilibrium.
Consider a system in thermal equilibrium 
with a heat bath at inverse temperature $\beta  \equiv  \frac{1}{ k_{\rm B} T}$,
wherein $k_{\rm B}$ denotes Boltzmann's constant.
We focus on classical systems for simplicity, though fluctuation relations have been extended to quantum systems~\cite{CampisiHT11}.
Suppose that a time-dependent external parameter $\lambda_t$
determines the system's Hamiltonian: $H = H( \lambda_t,  \mathbf{z} )$,
wherein $\mathbf{z}$ denotes a phase-space point.
If the system consists of an ideal gas in a box, 
$\lambda_t$ may denote the height of the piston that caps the gas.
Suppose that, at time $t = -\tau$, the system begins
with the equilibrium phase-space density
$e^{ - \beta H (\lambda_{-\tau},  \mathbf{z}  ) } / Z_{-\tau}$,
wherein the partition function $Z_{-\tau}$ normalizes the state.
The external parameter is then varied according to a predetermined schedule 
$\lambda_t$, from $t=-\tau$ to $t=\tau$.
The system evolves away from equilibrium if $\tau$ is finite.
In the gas example, the piston is lowered, compressing the gas.
We call this process the \emph{forward protocol}.

The \emph{reverse protocol} begins with the system at equilibrium 
relative to $H ( \lambda_\tau, \mathbf{z} )$. 
The external parameter is changed to $\lambda_{-\tau}$ 
along the time-reverse of the path followed during the forward protocol.
In the gas example, the piston is raised, and the gas expands.

Changing the external parameter
requires or outputs some amount of work.
We use the following sign convention:
The forward process tends to require an investment of 
a positive amount $W > 0$ of work,
and the reverse process tends to output $W > 0$.
The value of $W$ varies from trial to trial.
After performing many trials,
one can estimate the probability $P_\fwd(W)$ 
that any particular forward trial will cost an amount $W$ of work 
and the probability $P_\rev(-W)$
that any particular reverse trial will output an amount $W$.

These probabilities satisfy \emph{Crooks' Theorem}~\cite{Crooks99},
\begin{align}
\label{eq:CrooksThm}
   \frac{  P_\fwd(W) }{  P_\rev(-W)  }
   =  e^{ \beta ( W - \Delta F) }.
\end{align}
Here, $\Delta F  :=  F_\tau  -  F_{-\tau}$ denotes the difference between 
the free energy $F_\tau  =  - \beta^{-1} \log ( Z_{\tau} )$ of 
the Gibbs distribution 
$e^{ - \beta H (\lambda_{\tau},  \mathbf{z}  ) } / Z_{\tau}$
corresponding to the final Hamiltonian
and the free energy $F_{-\tau}  =  - \beta^{-1}  \log ( Z_{-\tau} )$ 
of the Gibbs distribution corresponding to $H( \lambda_{-\tau},  \mathbf{z} )$.
Multiplying each side of Crooks' Theorem by $P_\rev (-W)  e^{ \beta \Delta F}$,
then integrating over $W$,
yields a version of the \emph{nonequilibrium work relation}~\cite{Jarzynski97}:
\begin{align}
   e^{ \beta \Delta F }
   & =  \label{eq:JarzEqRev}
            \langle  e^{ \beta W } \rangle_\rev   \\
   & :=  \label{eq:Avg}
           \int_{ -\infty}^\infty  dW  \:  e^{ \beta W}  P_\rev (-W). 
\end{align}
The angle brackets denote an average over infinitely many trials. 
To calculate $\Delta F$, one performs many trials, estimates the average, 
and substitutes into Eq.~\eqref{eq:JarzEqRev}.

%
%
\emph{Thermal order-$\infty$ R\'enyi entropy ($H^\beta_\infty$)---}Entropies quantify uncertainties in statistical mechanics and in information theory.
Let $P  :=  \{ p_i \}$ denote a probability distribution
over a discrete random variable $X$.
The \emph{Shannon entropy}
$H_S(P)  :=  - \sum_i  p_i  \log (p_i)$
quantifies an average, over infinitely many trials, 
of the information one gains upon learning the value assumed by $X$ in one trial~\cite{CoverT12}.

$H_S$ has been generalized to a family of \emph{R\'enyi entropies} $H_\alpha$.
The parameter $\alpha  \in  [0,  \infty)$ is called the \emph{order}.
The $H_\alpha$'s quantify uncertainties related to finitely many trials.
In the limit as $\alpha \to \infty$, $H_\alpha$ approaches
\begin{align}
   H_\infty (P)  =  - \log ( p_{\rm max} ),
\end{align}
wherein $p_{\rm max}$ denotes the greatest $p_i$.
This maximal entropy has applications to randomness extraction:
The efficiency with which finitely many copies of $P$ 
can be converted into a uniformly random distribution 
$( \underbrace{  \frac{1}{d},   \ldots,   \frac{1}{d}  }_d  )$
is quantified with $H_\infty(P)$~\cite{RennerW04}.

The distributions $P_\fwd$ and $P_\rev$ in Crooks' Theorem are continuous.
Hence we need a continuous analog of $H_\infty$.
The definition
\begin{align}
   \label{eq:HInfDef}
   H^\beta_\infty (P)  
   :=  - \log ( p_{\rm max} / \beta )
\end{align}
has been shown to be useful in contexts that involve heat baths~\cite{YungerHalpernGDV15}.
$p_{\rm max}$ denotes the greatest value of the probability density $P$.
$p_{\rm max}$ can diverge, e.g., if $P$ represents a Dirac delta function.
But delta functions characterize the work distributions of quasistatic protocols,
whose work $W = \Delta F$ in every trial.
We focus on more-realistic, quick protocols.
$P_\fwd$ and $P_\rev$ are short and broad,
so $p_{\rm max}$ is finite.

The density $p_{\rm max}$ has dimensions of inverse energy,
which are canceled by the $\beta$ in Eq.~\eqref{eq:HInfDef}.
Hence the logarithm's argument is dimensionless.
For further discussion about $H^\beta_\infty$, see~\cite{YungerHalpernGDV15}.

%
%
%
%
\textbf{Quantification of dominance---}Let us return to 
the nonequilibrium work relation~\eqref{eq:Avg}. 
The exponential enlarges already-high $W$-values, which dominate the integral.
To estimate the integral accurately, 
one must perform trials that output large amounts of work.
Few trials do;
dominant $W$-values are \emph{atypical}~\cite{RareEvents}.
How many trials should one expect to need to perform, to achieve reasonable convergence of the exponential average in Eq.~\eqref{eq:Avg}?

An approximate answer was provided in~\cite{RareEvents}:
\begin{align}
\label{eq:NEstimate}
   N   \sim   e^{ \beta  ( \langle W \rangle_\fwd  -  \Delta F  ) },
\end{align}
wherein $\langle . \rangle_\fwd$ denotes an average with respect to $P_\fwd(W)$.
The \emph{average dissipated work} $\langle W \rangle_\fwd  -  \Delta F$
represents the mean amount of work wasted as heat.
Switching $\lambda_t$ quasistatically (infinitely slowly)
would cost an amount $\Delta F$ of work.
Switching at a finite speed costs more:
Work is dissipated into the bath as heat
when the system is driven away from equilibrium.
The dissipated work $W  -  \Delta F$
signifies the extra work paid to switch $\lambda_t$ 
in a finite amount of time.

How large must a $W$-value be to qualify as dominant?
This question remained open in~\cite{RareEvents}.
We propose a definition inspired by information-theoretic protocols 
in which an agent specifies an error tolerance.
The experimenter who switches $\lambda_t$, or the programmer who simulates trials,
chooses a threshold value of $w^\delta$
used to lower-bound the $W$-values considered large.
\begin{definition}
\label{definition:wdelta}
A work value $W$ extracted from a reverse-protocol trial 
is called $w^\delta$-\emph{dominant} if
$W \geq  w^\delta$
for the fixed value $w^\delta$ chosen by the agent. 
\end{definition}

A similar quantity is defined in~\cite{LuKofke_JChemPhys_01_I}.
Lu and Kofke assess the accuracy of free-energy-perturbation (FEP) calculations.
FEP is used to estimate free-energy differences $\Delta F$.
FEP results from a limit of nonequilibrium-fluctuation theory~\cite{Jarzynski97}.
In~\cite{LuKofke_JChemPhys_01_I}, a fixed-length simulation 
is assumed to be performed.
A difference $u$ between potential energies is measured.
$u$, in FEP, plays the role of $W$ in nonequilibrium fluctuation relations.
Lu and Kofke denote by $p(u)$ the probability
that a fixed-length simulation yields the potential-energy difference $u$.
\emph{Limit energies} $u_1$ and $u_2$ are defined as
the extreme realizable $u$-values.

Lu and Kofke fix the simulation length, 
then calculate the most likely limit energy, $u^*$.
In contrast, the agent in the present work fixes a tolerance $w^\delta$.
The number $N_\delta$ of required trials 
(similar to the simulation length) is then bounded.
Lu and Kofke also use the mode of $W^*$ 
to calculate the error in $\Delta F$.
The \emph{neglected-tail model} of~\cite{LuKofke_JChemPhys_01_I}
was extended from FEP to nonequilibrium fluctuation relations in~\cite{WuKofke_JChemPhys_04}.
When calculating the error in $\Delta F$,
Wu and Kofke average over possible values of the limit energy $W^*$.
The framework in~\cite{LuKofke_JChemPhys_01_I,WuKofke_JChemPhys_04} accommodates arbitrary $W^*$-values.
Yet statistical properties, such as the mean and mode,
are emphasized.
That emphasis is complemented by the present paper's 
information-theory-inspired choice of $w^\delta$ by the agent.
Additionally, the choice $w^\delta = \langle W \rangle_\fwd  -  \Delta F$
of the dissipated work is analyzed below.

Definition~\ref{definition:wdelta} enables us to bound the number $N_\delta$ 
of trials expected to be performed before 
one trial outputs a $w^\delta$-dominant amount of work.

%
%
%
%
\textbf{Bound on expected number $N_\delta$ of trials required---}Imagine 
implementing reverse trials until extracting 
a $w^\delta$-dominant amount of work from one trial.
One might have luck and extract $W \geq w^\delta$ on the first try.
But one would not expect to.
One would expect the number of trials to equal
the inverse $1 / \int_{w^\delta}^\infty  dW  \:  P_\rev ( W )$
of the probability that any particular reverse trial will output $W \geq w^\delta$.
In the notation of~\cite{YungerHalpernGDV15}, 
$\int_{w^\delta}^\infty  dW  \:    P_\rev ( W )  =  1 - \delta$ (see Fig.~\ref{fig:PRev}):
\begin{align}
\label{eq:NumProb}
   N_\delta  =  \frac {1 }{ 1 - \delta}.
\end{align}

Let us clarify what ``expect to perform $N_\delta$ trials'' means.
Imagine performing $M$ sets of reverse trials.
In each set, one performs trials until extracting $W \geq w^\delta$ from one trial.
Let $N_\delta^i$ denote the number of trials performed during the $i^{\rm th}$ set.
Consider averaging $N_\delta^i$ over the $M$ sets of trials:
$\frac{1}{M}  \sum_{i = 1}^M   N_\delta^i$.
As the number of sets grows large,
the average of the number of required trials in a set 
approaches the ``expected'' value $N_\delta$:
\begin{align}
   \lim_{ M \to \infty }
   \frac{1}{M}
   \sum_{i = 1}^M   N_\delta^i
   =   N_\delta.
\end{align}
This interpretation will facilitate our bounding of $N_\delta$.

%
%
\begin{figure}[tb]
\centering
\includegraphics[width=.42\textwidth, clip=true]{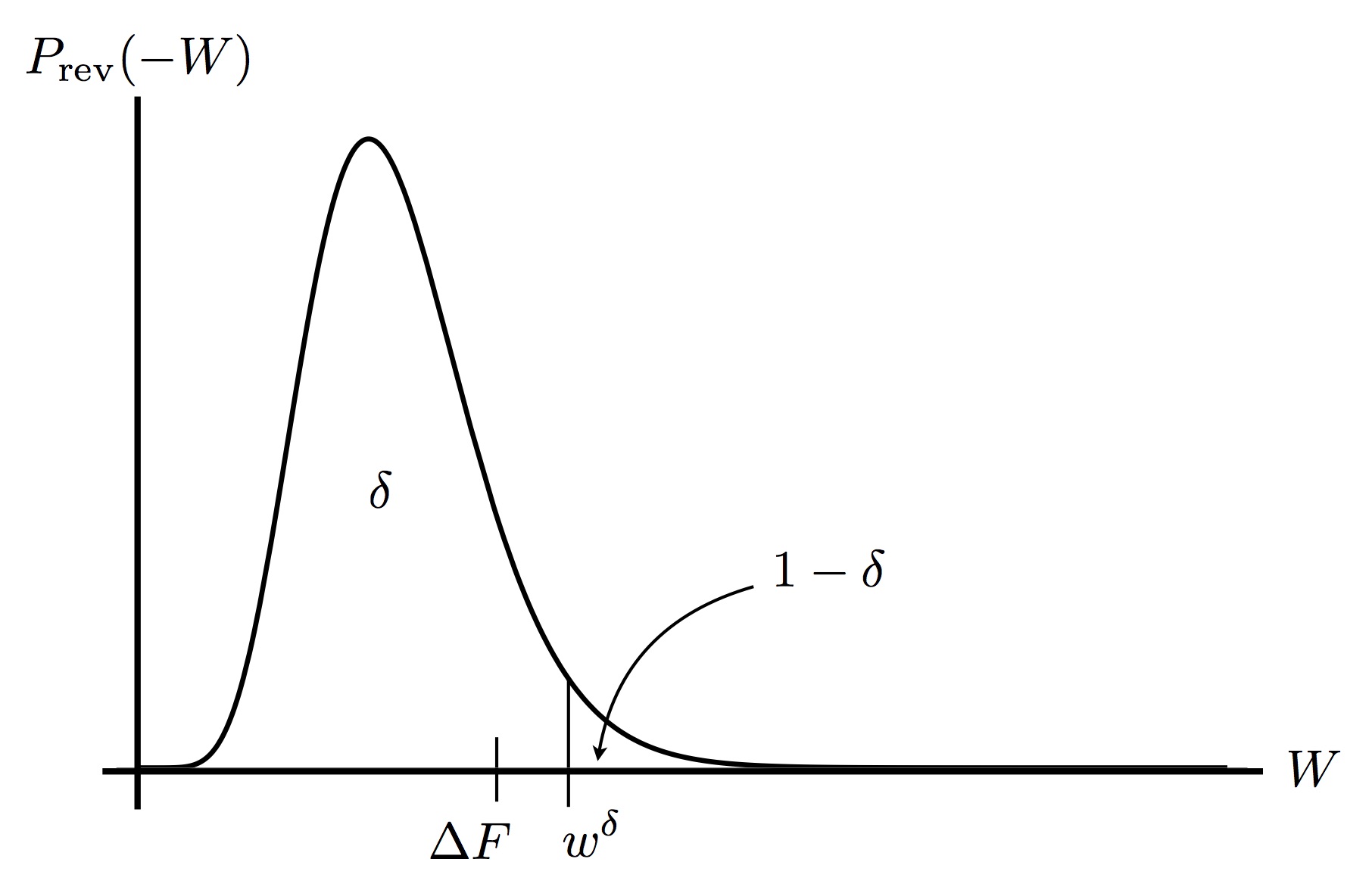}
\caption{\caphead{Dominant values of work extractable from reverse-protocol trials:}
Large values $W$ of work contribute the most to 
the integral in the nonequilibrium fluctuation relation~\eqref{eq:JarzEqRev}.
An amount $W$ of extracted work is called \emph{$w^\delta$-dominant} if 
it is at least as great as the threshold $w^\delta$ specified by the experimenter:
$W \geq w^\delta$.
The probability that any particular reverse trial 
will output a $w^\delta$-dominant amount of work is 
$\int_{ w^\delta}^\infty  dW  \:  P_\rev ( -W )   =  1  -  \delta$.
This probability equals the area of the region under the distribution's right-hand tail.}
\label{fig:PRev}
\end{figure}

%
%
\begin{theorem}[Bound on expected number of trials]
\label{theorem:NBound}
The number $N_\delta$ of reverse trials 
expected to be performed
before one trial outputs a $w^\delta$-dominant amount
$W \geq w^\delta$ of work 
is bounded as
\begin{align}   
   \label{eq:NBound}
   N_\delta    \geq
   e^{ \beta (w^\delta  -  \Delta F)  +  H^\beta_\infty ( P_\fwd ) }.
\end{align}
\end{theorem}

\begin{proof}
The inequality
\begin{align}
\label{eq:wdBound}
   w^\delta   \leq   \Delta F
   -  \frac{1}{\beta}  [
   H^\beta_\infty ( P_\fwd )   +   \log ( 1  -  \delta ) ]
\end{align}
was derived in~\cite{YungerHalpernGDV15}.
The derivation relies on the definitions of $1 - \delta$ and $H^\beta_\infty$, 
on Crooks' Theorem, and on the bound $P_\fwd(W)  \leq  p_{\rm max}  \;  \forall \, W$.
Solving for $1 - \delta$, then inverting the probability [Eq.~\eqref{eq:NumProb}], yields Ineq.~\eqref{eq:NBound}.

\end{proof}

%
%
Inequality~\eqref{eq:NBound} implies that the bound on $N_\delta$ increases with $w^\delta$, which makes sense.
As we raise the threshold $w^\delta$, fewer work values qualify as $w^\delta$-dominant.
Hence more trials are expected to be required 
before a $w^\delta$-dominant work value is observed.

%
%
\emph{Improvement over Relation~\eqref{eq:NEstimate}---}Inequality~\eqref{eq:NBound} resembles its inspiration, 
Relation~\eqref{eq:NEstimate}, which states that the number $N$ of trials required
to achieve convergence of the average in Eq.~(\ref{eq:Avg}) increases exponentially with the
average dissipated work $\langle W \rangle_{\rm fwd}  -  \Delta F$.
Similarly, the bound on $N_\delta$ increases exponentially with 
the ``one-shot dissipated work'' $w^\delta  -  \Delta F$.
This $w^\delta  -  \Delta F$ represents
the work sacrificed for time
in a forward trial that costs an amount $w^\delta$ of work.

Moreover, $N_\delta$ is defined in terms of the reverse process.
Yet the bound on $N_\delta$ given by
Ineq.~\eqref{eq:NBound} depends on the forward work distribution, via $H^\beta_\infty ( P_\fwd )$.
Similarly, in Relation~\eqref{eq:NEstimate}, the number $N$ of repetitions of the reverse process
required for the convergence of Eq.~(\ref{eq:Avg}) depends on the forward work distribution 
$P_\fwd(W)$, via $\langle W \rangle_{\rm fwd}$.

Despite its similarity to Relation~\eqref{eq:NEstimate}, 
Ineq.~\eqref{eq:NBound} offers three advantages. 
First, Ineq.~\eqref{eq:NBound} quantifies dominance with $\delta$, reflecting the agent's accuracy tolerance.
Next,  Relation~\eqref{eq:NEstimate} is a rough estimate.
Inequality~\eqref{eq:NBound} is a strict bound
on the number of trials expected to be performed 
before a $w^\delta$-dominant amount of work is extracted.
Finally, Ineq.~\eqref{eq:NBound} contains an entropy 
that has no analog in Relation~\eqref{eq:NEstimate}. 
The entropy tightens the bound when 
\begin{align}
\label{eq:Tighten}
   p_{\rm max}  <  \beta.
\end{align}
This inequality is satisfied, for instance, in RNA-hairpin experiments used to test fluctuation theorems~\cite{CollinRJSTB05}.

To appreciate these advantages over Relation~\eqref{eq:NEstimate},
we can define $w^\delta$-dominant work values 
by choosing $w^\delta  =  \langle W  \rangle_\fwd$,
as in~\cite{RareEvents}.
The bound becomes
\begin{align}
   N_\delta   \geq   
   e^{ \beta ( \langle W  \rangle_\fwd  -  \Delta F )  +  H^\beta_\infty (P_\fwd) }.
\end{align}
When $p_{\rm max}  <  \beta$ (such that $H^\beta_\infty>1$), 
the number of trials required for Eq.~\eqref{eq:Avg} to converge 
exceeds the prediction in Relation~\eqref{eq:NEstimate}.


We can gain further insight by rewriting Ineq.~\eqref{eq:NBound} as
\begin{align}
   N_\delta   \geq
   \frac{\beta}{ p_{\rm max} }
   e^{ \beta ( \langle W  \rangle_\fwd  -  \Delta F ) },
   \label{eq:Nbound_rewritten}
\end{align}
using the definition of $H^\beta_\infty$ [Eq.~\eqref{eq:HInfDef}].
The fraction ${\beta} / p_{\rm max}$ represents approximately
the number of forward trials performed 
before one trial's $W$-value
falls within a width-$( \kB T )$ window
about the most probable work value $W_{\rm max}$:
$W  \in  [ W_{\rm max}  -  \frac{ \kB T }{2},  W_{\rm max}  +  \frac{ \kB T }{2} ]$.
The value of ${\beta} / p_{\rm max}$ generically increases with the width of the distribution $P_\fwd(W)$.
Hence the bound on $N_\delta$, as written in Ineq.~(\ref{eq:Nbound_rewritten}), is a product of two factors. 
The first depends on the forward work distribution's width;
and the second, on its mean.
In contrast, Relation~\eqref{eq:NEstimate} depends only on the mean.

The area under distributions' tails is evoked also in~\cite{WuKofke_JChemPhys_04}.
Wu and Kofke use their \emph{neglected tail model}
to estimate the bias in $\Delta F$.

%
%
\emph{Classical vs. quantum applications---}Classical mechanics describes
most experiments and numerical simulations for which $N_\delta$ needs calculating.
Nonetheless, quantum experiments merit consideration.

We have assumed that the work distributions 
$P_\fwd(W)$ and $P_\rev(-W)$ are continuous.
Classical systems have continuous work distributions:
A classical system's possible energies form a continuous set.
So do the differences between possible energy values---the 
possible work values.
Continuousness leads to Ineq.~\eqref{eq:wdBound},
from which Theorem~\ref{theorem:NBound} is derived.
How to extend Ineq.~\eqref{eq:wdBound} to discrete sets
of possible work values is unclear.

Quasiclassical systems can have continuous work distributions.
By \emph{quasiclassical}, we mean
systems whose energies form a discrete set
but whose states (density operators) 
commute with the Hamiltonian.
Consider a quasiclassical system 
that exchanges heat with a bath throughout the work extraction.
The system always occupies an energy eigenstate
if the energy is measured frequently~\cite{QuanD08,YungerHalpernGDV15}.
The work performance lowers the system's energy levels.
Suppose that two levels fall at different rates.
The system can hop from level to level at any time.
Hopping at time $t$ can output infinitesimally more work
than hopping at time $t + dt$~\cite{YungerHalpernGDV15}.
Such quasiclassical systems obey Theorem~\ref{theorem:NBound}.

Discrete work distributions characterize quantum systems
that undergo the \emph{two-time-measurement protocol}~\cite{Tasaki00,Kurchan00}.
A quantum system undergoes an energy measurement,
is isolated from the bath,
performs work unitarily,
and suffers another energy measurement.
The differences between the possible measurement outcomes
form a discrete set.
Extending Theorem~\ref{theorem:NBound} to such protocols
could merit investigation.
One might incorporate the bin width of the histograms
used to approximate $P_\fwd(W)$ and $P_\rev(-W)$.
On the other hand, bin widths are artificial approximation tools,
chosen by the experimenter.
One might prefer a theory independent of such an approximation~\cite{YungerHalpernGDV15}.
Extensions may be galvanized 
by the evolution of quantum experiments
to a point that requires $N_\delta$ estimations.

%
%
\emph{Fail safety---}Fail safety is a property of 
certain estimates calculated from incomplete data.
The bound on $N_\delta$ depends on the free-energy difference $\Delta F$.
$\Delta F$ is estimated from forward-trial data.
Finitely many forward trials are performed.
Hence the $\Delta F$ estimate is biased.
This bias skews one's estimate of the $N_\delta$ bound.
Suppose that the estimate lay above the true value of $N_\delta$.
The $N_\delta$-bound estimate would lead the agent
to perform enough trials
to estimate $\Delta F$ with reasonable accuracy.
The $N_\delta$-bound estimate would be \emph{fail-safe}~\cite{WuKofke_JChemPhys_2005_I,WuKofke_JChemPhys_2005_II}.
Fail safety is often desirable.
Surprisingly, a lack of fail safety benefits Theorem~\ref{theorem:NBound},
because Ineq.~\eqref{eq:NBound} lower-bounds $N_\delta$.

The bias in the $\Delta F$ estimate
lowers estimates of the $N_\delta$ bound
below the bound's true value:
The nonequilibrium fluctuation relation can be expressed as
$e^{ - \beta \Delta F }  =  \langle e^{ - \beta W } \rangle_\fwd$~\cite{Jarzynski97}.
Solving for $\Delta F$ yields
\begin{align}
   \label{eq:DeltaFJarz}
   \Delta F  =  - \frac{1}{ \beta }
   \log \langle e^{ - \beta W } \rangle_\fwd.
\end{align}
Forward trials tend to cost large amounts of work:
Typical $W$-values are high.
High $W$-values lower the estimate of
$\langle e^{ - \beta W } \rangle_\fwd$
below the average's true value.
This low estimate raises the $\Delta F$ estimate above the true $\Delta F$ value,
by Eq.~\eqref{eq:DeltaFJarz}.
This overestimate of $\Delta F$ 
lowers the estimate of the $N_\delta$ bound
below the bound's true value, by Ineq.~\eqref{eq:NBound}.

In summary, Ineq.~\eqref{eq:NBound} lower-bounds $N_\delta$.
Estimating this lower-bound with biased data
generates an even lower bound on $N_\delta$:
\begin{align}
   N_\delta \geq {\rm (True  \;  lower   \;  bound) }
   \geq   \rm{ (Estimated   \;   lower   \;   bound). }
\end{align}
This second lower-bounding 
renders Theorem~\ref{theorem:NBound} robust
against the bias in the $\Delta F$ estimate.

This robustness precludes fail-safety.
Suppose that the protocol were fail-safe.
The estimate of the $N_\delta$ bound 
would lie above the true bound:
\begin{align}
   N_\delta  \geq  {\rm (True  \;  lower   \;  bound) }
   \leq    \rm{ (Estimated   \;   lower   \;   bound). }
\end{align}
One's estimate of the lower bound on $N_\delta$
\emph{would not necessarily lower-bound $N_\delta$.}
An experimentalist could not use Theorem~\ref{theorem:NBound}.
The theorem benefits, unusually, from a lack of fail-safety.

%
%
%
%
\textbf{Evaluating the $N_\delta$ bound---}Not only does Ineq.~\eqref{eq:NBound} 
have a theoretically satisfying form,
but it can also be estimated in practice.
We will discuss how to estimate 
the $H^\beta_\infty ( P_\fwd )$ and the $\Delta F$ in the bound.
The bound can be estimated reasonably, we argue,
from not too many trials.

The experimental set-up determines $\beta$, and the agent chooses $w^\delta$.
$H^\beta_\infty ( P_\fwd )$ and $\Delta F$ can be estimated
if one implements the ``good practice'' of bidirectionality.
To mitigate errors in $(\Delta F)$-estimates,
one should perform forward trials, perform reverse trials,
and combine all the data~\cite{GoodPractices}.
Upon performing several forward trials, one can estimate 
$H^\beta_\infty ( P_\fwd )$ and $\Delta F$.
One can estimate the $N_\delta$ bound, then
perform (probably at least $N_\delta$) reverse trials
until observing a $w^\delta$-dominant work value, 
and improve the $(\Delta F)$-estimate.\footnote{
$N_\delta$ can be estimated from reverse trials alone, less reliably.
One could perform a few reverse trials, estimate $P_\rev(-W)$, and estimate $\Delta F$.
From these estimates
and from Crooks' Theorem, one could estimate $P_\fwd(W)$.
From $P_\fwd(W)$, one could estimate $H_\infty^\beta ( P_\fwd )$, 
then estimate the $N_\delta$ bound. 
One could repeat this process, improving one's estimate of the bound, 
until observing a $w^\delta$-dominant work value.
But the estimate of $\Delta F$ is expected to jump repeatedly~\cite{RareEvents}.
This sawtooth behavior, 
as well as the piling of estimate upon estimate, 
may taint the estimates of the bound.}

$H^\beta_\infty$ depends on $p_{\rm max}$,
the greatest probability (per unit energy) of any possible forward-trial outcome.
This outcome will likely appear in many trials.
Hence one expects to estimate $H^\beta_\infty$ well
from finitely many forward trials. 

%
%
\begin{figure}[tb]
\centering
\includegraphics[width=.42\textwidth, clip=true]{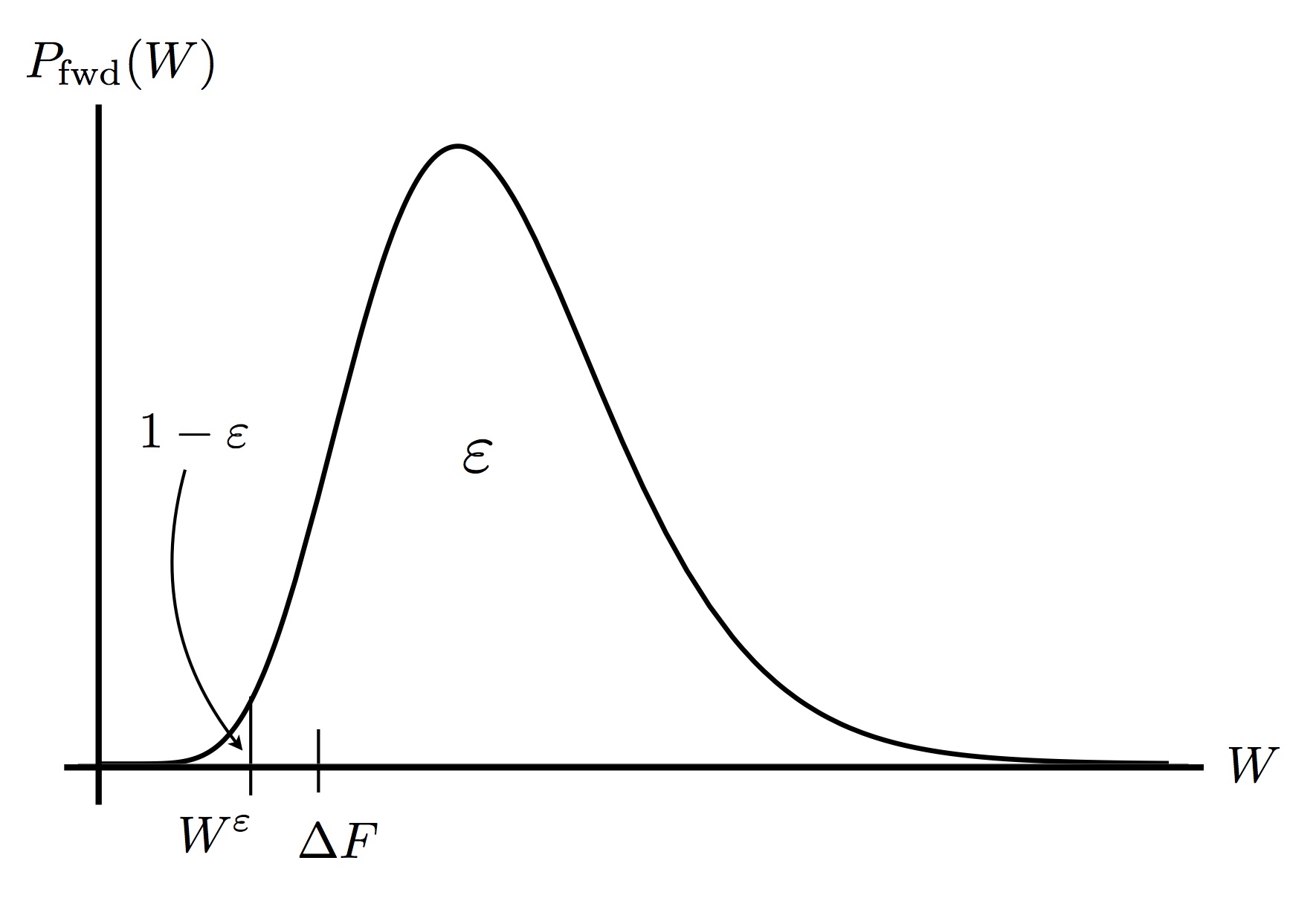}
\caption{\caphead{Dominant values of work invested in forward-protocol trials:}
Small values $W$ of work dominate the nonequilibrium work relation~\eqref{eq:JarzEq}.
An amount $W$ of invested work is called \emph{$W^\varepsilon$-dominant} if 
it lies below or on the threshold $W^\varepsilon$ chosen by the experimenter:
$W \leq W^\varepsilon$.
The probability that any particular forward trial will require 
a $W^\varepsilon$-dominant amount of work is 
$1  -  \varepsilon  =  \int_{-\infty}^{ W^\varepsilon }  dW  \:  P_\fwd (W)$.
This probability equals the area under the distribution's left-hand tail.}
\label{fig:PFwd3}
\end{figure}
%
%
%

%
%
%
%
\textbf{Forward-protocol bound---}Trials or computations performed in one direction
can cost more time than 
trials or computations performed in the opposite direction~\cite{LuKofke_JChemPhys_01_I}.
We have bounded a number $N_\delta$ of reverse trials.
Similarly, we should bound the number $N_\varepsilon$ of forward trials 
expected to be performed
before a $W^\varepsilon$-dominant amount of work is invested.
The analysis is analogous to that of $N_\delta$.

The nonequilibrium work relation for the forward process is
\begin{align}
\label{eq:JarzEq}
   \langle e^{ - \beta W } \rangle_\fwd  =  e^{ - \beta \Delta F }.
\end{align}
The forward trials that dominate the average in Eq.~\eqref{eq:JarzEq}
cost unusually small amounts of work.
In the notation of~\cite{YungerHalpernGDV15},
$W^\varepsilon$-dominant work values satisfy $W  \leq  W^\varepsilon$,
for a tolerance $W^\varepsilon$ chosen by the agent.
Each forward trial has a probability $1 - \varepsilon$ of costing
a $W^\varepsilon$-dominant amount of work (see Fig.~\ref{fig:PFwd3}).
Theorem 4 of~\cite{YungerHalpernGDV15} 
bounds $W^\varepsilon$ in terms of $1 - \varepsilon$.
Solving for $1 - \varepsilon$, then inverting,
bounds the number $N_\varepsilon  =  1 / (1 - \varepsilon)$ 
of forward trials expected to be performed
before any trial costs a $W^\varepsilon$-dominant amount of work:
\begin{align}
   N_\varepsilon   \geq
   e^{ - \beta ( W^\varepsilon  -  \Delta F)  +  H^\beta_\infty ( P_\rev ) }.
\end{align}

%
%
%
%
\textbf{Error estimate:} 
Calculating the error in a $(\Delta F)$-estimate
is crucial but difficult.
Whenever one infers a value from data,
the inference's reliability must be reported.
Common error analyses do not suit estimates of $(\Delta F)$-values,
for two reasons.
First, $\Delta F$ depends on the random variable $W$ logarithmically
[see Eq.~\eqref{eq:JarzEqRev}].
Second, $W$ tends not to be Gaussian.
Approaches such as
an uncontrolled approximation, in the form of a truncation of a series expansion,
have been proposed~\cite{GoodPractices}.
Our approach centers on the agent's choice of $w^\delta$.

Consider choosing a $w^\delta$-value and performing $N_\delta$ trials.
With what accuracy can one estimate $\Delta F$?
We will bound the percent error
\begin{align}
   \label{eq:PercErrDef}
   \epsilon  :=  \left\lvert   \frac{ \Delta F  -  (\Delta F)_\est  }{   \Delta F  }   \right\rvert
\end{align}
roughly.
To render the problem tractable,
we assume that one knows
the exact form of $P_\rev (-W)$ for all $W \leq w^\delta$.

This assumption features also in the neglected-tail model of~\cite{LuKofke_JChemPhys_01_I,LuKofke_JChemPhys_01_II,WuKofke_JChemPhys_04}.
The percent error in $e^{ - \beta \Delta F }$ is calculated,
with free-energy perturbation theory (FEP), in~\cite{LuKofke_JChemPhys_01_I}.
This percent error, if small, approximates the absolute error
$\Delta F  -  (\Delta F)_\est$ in the free-energy difference~\cite{LuKofke_JChemPhys_01_II}.
Bias calculations are extended from FEP
to nonequilibrium work fluctuation relations in~\cite{WuKofke_JChemPhys_04}.

%
%
\begin{theorem}[Approximate error bound]
   Let the work tolerance be $w^\delta \in (- \infty, \infty )$.
   Let $(\Delta F)_\est$ denote the estimate of the free-energy difference $\Delta F$
   inferred from data taken during $N_\delta$ trials.
   If $(\Delta F)_\est$ is calculated from the exact form of 
   $P_\rev (-W)   \quad \forall  \:  W \leq w^\delta$,
   the estimate has a percent error of
   \begin{align}
      \label{eq:ErrorBound}
       \epsilon   \geq   \frac{1}{ \beta (\Delta F) }
      \Big[ \eta  +  O (\eta^2) \Big], 
   \end{align}
   wherein
   \begin{align}
      \eta  :=  \frac{ e^{ \beta w^\delta } }{  N_\delta  \langle e^{\beta W} \rangle_\rev }.
   \end{align}
\end{theorem}

%
%
\begin{proof}
Let us solve the nonequilibrium work relation~\eqref{eq:JarzEqRev} for $\Delta F$:
\begin{align}
   \label{eq:Error1}
   \Delta F  
   & =  \frac{1}{\beta}  \log  \Big(  
   \langle  e^{\beta W}  \rangle_\rev   \Big)  \\
   & =  \frac{1}{\beta}  \log  \left( 
          \int_{ -\infty }^\infty  dW  \:  e^{ \beta W }  P_\rev (-W)  \right).
\end{align}
The estimate has a similar form:
\begin{align}
   ( \Delta F)_\est
   & =  \frac{1}{\beta}  \log  \left( 
          \int_{ -\infty}^{ w^\delta }   dW  \:  e^{ \beta W }  P_\rev (-W)  \right)  \\
   & =  \frac{1}{\beta}  \log  \Bigg( 
          \int_{ -\infty}^\infty   dW  \:  e^{ \beta W }  P_\rev (-W)
          \nonumber  \\  &  \qquad
          -   \int_{ w^\delta }^\infty  dW  \:  e^{ \beta W }  P_\rev (-W)
            \Bigg).
\end{align}
We replace the first integral with $\langle  e^{\beta W}  \rangle_\rev$,
using Eq.~\eqref{eq:Error1}.
The second term, representing the error,
is expected to be much smaller than the first term.
This second term will serve as a small parameter in a Taylor expansion:
\begin{align}
   ( \Delta F)_\est
   & =  \frac{1}{\beta} \Bigg[ 
          \log  \Big(   \langle  e^{\beta W}  \rangle_\rev   \Big)
          \nonumber  \\ &  
          +  \log   \Bigg(  1  - 
                        \frac{   \int_{ w^\delta }^\infty  dW  \:  e^{ \beta W }  P_\rev (-W)   }{
                        \langle  e^{\beta W}  \rangle_\rev }   \Bigg)
          \Bigg]   \\
   & =  \label{eq:Error2}
           \Delta F   
           -    \frac{1}{\beta}   \Big\{  \eta'  +  O (  \, [\eta']^2  \, )   \Big\},
\end{align}
wherein
\begin{align}
   \label{eq:EtaPrime}
   \eta'  :=  \frac{   \int_{ w^\delta }^\infty  dW  \:  e^{ \beta W }  P_\rev (-W)   }{
                        \langle  e^{\beta W}  \rangle_\rev }.
\end{align}

We can bound the numerator, using Fig.~\ref{fig:PRev}:
\begin{align}
   \int_{ w^\delta }^\infty  &  dW  \:   e^{ \beta W }  P_\rev (-W)
   \\ & \geq   e^{ \beta w^\delta }
                \int_{ w^\delta }^\infty  dW  \:  P_\rev (-W)   \\
   & =  e^{ \beta w^\delta }  (1 - \delta)  
   =  \frac{  e^{ \beta w^\delta }  }{  N_\delta  }.
\end{align}
Substituting into Eq.~\eqref{eq:EtaPrime} yields
$\eta'  \geq  \eta$.
Hence Eq.~\eqref{eq:Error2} reduces to
\begin{align}
   ( \Delta F)_\est
   \leq   \Delta F   
   -    \frac{1}{\beta}   \Big[  \eta  +  O ( \, \eta^2 \, )   \Big].
\end{align}
Substituting into the percent error's definition [Eq.~\eqref{eq:PercErrDef}]
yields Ineq.~\eqref{eq:ErrorBound}. 
\end{proof}

The approximate error bound can be estimated 
from agent-chosen parameters and from data:
The experiment's set-up determines the value of $\beta$.
The agent chooses the value of $w^\delta$. 
For $N_\delta$, one can substitute the number of trials performed
[or can substitute from Ineq.~\eqref{eq:NBound}].
$\Delta F$ and $\langle e^{ \beta W } \rangle_\rev$
can be estimated from data.

%
%
%
%
\textbf{Numerical experiments---}To illustrate our analytical results,
we considered the weakly interacting dilute classical gas.
This system's forward and reverse work distributions can be calculated exactly~\cite{CrooksJarz}. 
The gas begins in equilibrium with a heat bath 
at inverse temperature $\beta  \equiv  \frac{1}{ \kB T}$.
During the forward protocol, the gas is isolated from the bath at $t = -\tau$.
The gas is quasistatically compressed,
its temperature rising from $T$.
During the reverse protocol, the gas expands and cools.
When discussing either direction,
we denote the initial volume by $V_0$ and the final volume by $V_1$.

%
%
\begin{figure}[tb]
\centering
\includegraphics[width=.5\textwidth, clip=true]{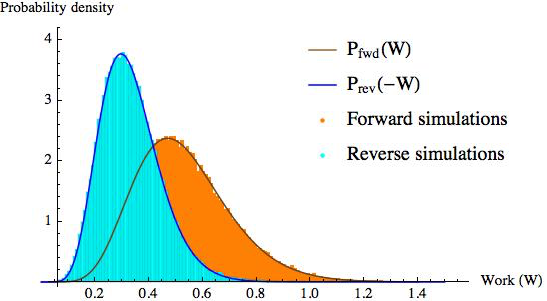}
\caption{
\caphead{Probability densities and numerical data 
for a weakly interacting dilute gas:}
(Color online.)
We considered a gas undergoing compression (a forward protocol) 
and expansion (a reverse protocol).
The probability per unit energy
that any particular trial will involve an amount $W$ of work [Eq.~\eqref{eq:CrooksJarz}]
was calculated in~\cite{CrooksJarz}.
The short, right-hand, brown curve represents $P_\fwd(W)$.
The tall, left-hand, dark-blue curve represents $P_\rev(-W)$.
By sampling work values from these distributions, we effectively
simulated each protocol $10^5$ times.
The cyan bars (under the left-hand curve) depict the data gathered from the forward-protocol samples.
The orange bars (under the right-hand curve) depict the data from the reverse-protocol samples.
}
\label{fig:PPlots}
\end{figure}

The probability densities over the possible work values
were calculated in~\cite{CrooksJarz}:
\begin{align}
\label{eq:CrooksJarz}
   P(W)  =  \frac{ \beta }{ | \alpha |  \Gamma(k) }
   \left(  \frac{ \beta W }{ \alpha }   \right)^{k - 1} 
   e^{ - \beta W / \alpha } \:
   \theta( \alpha W ).
\end{align}
During the forward protocol, 
$\alpha  :=  ( V_0 / V_1 )^{2 / 3}  -  1  >  0$;
during the reverse, $\alpha < 0$.
The gamma function is denoted by $\Gamma(k)$;
and its argument, by $k  :=  \frac{3}{2} n$,
wherein $n$ denotes the number of particles.
The theta function $\theta ( \alpha W )$
ensures that $W \geq 0$ is invested in forward trials
(for which $P  =  P_\fwd$);
and $W \leq 0$, in reverse trials (for which $P = P_\rev$).

This model illustrates accuracies also in~\cite{Kofke06}.
Kofke synthesizes theoretical results about $\Delta F$ estimates.
Relevant results include the neglected-tail model~\cite{WuKofke_JChemPhys_04}.
Numerical experiments on the gas illustrate those results.

%
%
\begin{figure}[tb]
\centering
\includegraphics[width=.45\textwidth, clip=true]{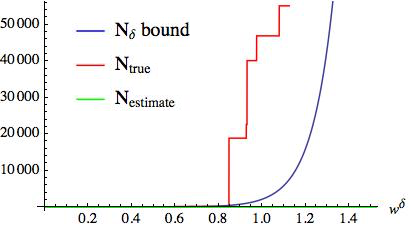}
\caption{\caphead{Three number-of-required-trial measures:}
(Color online.)
The abscissa shows possible choices of 
the threshold $w^\delta$ for $w^\delta$-dominant work values.
The blue (gently sloping) curve, calculated from $10^5$ forward-trial samples,
represents the bound on the number $N_\delta$ of reverse trials
expected to be performed before 
any trial outputs a $w^\delta$-dominant amount 
$W \geq w^\delta$ of work (Theorem~\ref{theorem:NBound}).
The red (staggered) curve, calculated from $10^5$ reverse-trial samples,
depicts the actual number $N_{\rm true}$ of trials performed before
$W \geq w^\delta$ is extracted.
The green curve (flat, nearly coincident with the abscissa)
was calculated from forward-trial samples.
This green curve represents Relation~\eqref{eq:NEstimate}:
an estimate $N_{\rm est}$ of the number of trials required 
to extract a dominant amount of work,
wherein the meaning of ``dominant'' is unspecified.
The blue (gently sloping) curve follows the red (staggered) curve's shape 
more faithfully than the green (flat) does,
illustrating the precision of Theorem~\ref{theorem:NBound}.
As expected, the blue (gently sloping) curve 
lower-bounds the red (staggered) at most $w^\delta$-values.
}
\label{fig:Simulation}
\end{figure}

We sampled $10^5$ values of $W$ from the forward (compression) work distribution 
and $10^5$ values from the reverse (expansion) work distribution. 
Figure~\ref{fig:PPlots} shows the probability densities and the sampled data.
We chose $V_0 / V_1  =  2$ and $n = 6$, following~\cite{CrooksJarz}, and $\beta = 10$.
Dividing a histogram of the forward-protocol data into 50 bins yielded $p_{\rm max}  =  1.577$. 
Satisfying Ineq.~\eqref{eq:Tighten}, this $p_{\rm max}$
enables $H^\beta_\infty (P_\fwd)$ to tighten the $N_\delta$ bound.

Figure~\ref{fig:Simulation} illustrates our results.
Possible values of $w^\delta$ appear along the abscissa.
The blue (gently sloping) curve shows the $N_\delta$ bound,
calculated from forward-trial samples, in Theorem~\ref{theorem:NBound}.
The red (staggered) curve, calculated from reverse-trial samples,
shows after how many reverse trials ($N_{\rm true}$)
$W \geq w^\delta$ was extracted during one trial.
$N_{\rm true}$ has a jagged, step-like shape, as one might expect.

The green curve (flat) lies close to the abscissa.
This curve depicts the estimate, in~\cite{RareEvents}, 
of the number of reverse trials expected to be performed
before one trial outputs a dominant work value,
for an unspecified meaning of ``dominant.''
We calculated $N_{\rm est}  =  3$
by simulating forward trials,
calculating the average dissipated work, 
and substituting into Relation~\eqref{eq:NEstimate}.

%
%
\begin{figure}[tb]
\centering
\includegraphics[width=.45\textwidth, clip=true]{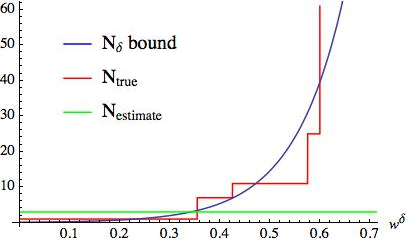}
\caption{
\caphead{Three number-of-required-trial measures
at low threshold work values $w^\delta$:}
(Color online.)
At most threshold values $w^\delta$, 
the $N_\delta$ bound (blue, gently sloping) lower-bounds 
the actual number $N_{\rm true}$ (red, staggered)
of reverse trials performed before 
any trial outputs a $w^\delta$-dominant amount $W \geq w^\delta$ of work.
At low $w^\delta$-values, the red curve zigzags across the blue (gently sloping).
This zigzagging stems from the technical definition of $N_\delta$.
}
\label{fig:ZoomPlot}
\end{figure}

The curves' shapes and locations illustrate the $N_\delta$ bound's advantages.
The bound (the blue, gently sloping curve) 
hugs the actual number $N_{\rm true}$ of trials required
(the red, staggered curve) more closely than $N_{\rm est}$ 
(the green, flat curve) does.
$N_{\rm est}$ remains flat, 
whereas the $N_\delta$ bound rises as $N_{\rm true}$ rises.
The $N_\delta$ bound often lower-bounds $N_{\rm true}$, as expected.
When $w^\delta$ is small, the $N_\delta$ bound weaves above and below $N_{\rm true}$, as shown in Fig.~\ref{fig:ZoomPlot}.
The reason was explained above Theorem~\ref{theorem:NBound}:
$N_\delta$ denotes the number of trials \emph{expected},
in a sense defined by probability and frequency,
to be required.
One might get lucky and extract $W \geq w^\delta$
before performing $N_\delta$ trials.
The dropping of the $N_{\rm true}$ curve below the $N_\delta$ bound
represents such luck.
But one expects to perform $N_\delta$ trials, 
and the $N_\delta$ bound lower-bounds $N_{\rm true}$ for most $w^\delta$-values.

%
%
%
%
\textbf{Conclusions---}We have sharpened predictions
about the number of experimental trials required to estimate $\Delta F$
from fluctuation relations.
We improved the approximation in~\cite{RareEvents} to an inequality,
tightened the bound (in scenarios of interest) 
with an entropy $H^\beta_\infty$,
freed the experimenter to choose a tolerance $w^\delta$ for dominance,
and approximately bounded the error in an estimate of $\Delta F$.
How to choose $w^\delta$ merits further investigation.
We wish to be able to specify the greatest error $\epsilon$
acceptable in an estimate of $\Delta F$.
From $\epsilon$, we wish to infer the number $N^\epsilon$ of trials
we should expect to perform.
This entire investigation improves the rigor with which free-energy differences $\Delta F$
can be estimated from experimental and numerical-simulation data.

%
%
%
%
\textbf{Acknowledgements---}NYH thanks Yi-Kai Liu 
for conversations about error probability
and thanks Alexey Gorshkov for hospitality at QuICS.
Part of this research was conducted while NYH was visiting 
the QuICS and the UMD Department of Chemistry and Biochemistry.
NYH was supported by an IQIM Fellowship and NSF grant PHY-0803371. The Institute for Quantum
Information and Matter (IQIM) is an NSF Physics Frontiers Center supported by the Gordon
and Betty Moore Foundation.
CJ was supported by NSF grant DMR-1506969.

%
%
%

\bibliographystyle{h-physrev}
\bibliography{Number_trials_bib}


\end{document}